\documentclass{article}
\usepackage{amsmath}
\usepackage{amssymb}
\usepackage{amsthm}
\usepackage{graphicx}
\usepackage{verbatim}
\usepackage{xcolor}
\usepackage[caption=false]{subfig}
\usepackage{listings}
\usepackage{hyperref}
\usepackage{tikz}
\usepackage{pgfplots} 
\usepackage{float}
\usepackage{appendix}
\usepackage{authblk}
\usepackage{hyphenat}
\usepackage[english]{babel}

\usepackage{algorithm}
\usepackage{algpseudocode}

\makeatletter
\newenvironment{breakablealgorithm}
  {
   \begin{center}
     \refstepcounter{algorithm}
     \hrule height.8pt depth0pt \kern2pt
     \renewcommand{\caption}[2][\relax]{
       {\raggedright\textbf{\fname@algorithm~\thealgorithm} ##2\par}%
       \ifx\relax##1\relax 
         \addcontentsline{loa}{algorithm}{\protect\numberline{\thealgorithm}##2}%
       \else 
         \addcontentsline{loa}{algorithm}{\protect\numberline{\thealgorithm}##1}%
       \fi
       \kern2pt\hrule\kern2pt
     }
  }{
     \kern2pt\hrule\relax
   \end{center}
  }
\makeatother

\pgfplotsset{compat=1.18} 

\newcommand{\code}[1]{\texttt{#1}}
\newcommand{\system}[1]{\textsc{#1}}
\newcommand{\eofex}{\hfill$\blacksquare$}

\newcommand{\union}{\cup}

\newcommand{\set}[1]{\{#1\}}

\newcommand{\rg}[3]{{#1}\,{#2}\ldots{#2}\,{#3}}
\newcommand{\eset}[2]{\set{\rg{#1}{,}{#2}}}

\newcommand{\lequiv}{\leftrightarrow}

\newcommand*{\PL}{\mathrm{PL}}
\newcommand*{\size}{size}
\newcommand*{\Nset}{\mathbb{N}}
\newcommand*{\err}{\mathrm{err}}

\DeclareMathOperator{\DDelta}{\Delta}

\newtheorem{theorem}{Theorem}
\newtheorem{proposition}[theorem]{Proposition}
\newtheorem{lemma}[theorem]{Lemma}

\theoremstyle{definition}

\newtheorem{example}[theorem]{Example}

\date{}

\begin{document}
\title{Short Boolean Formulas as  Explanations in Practice
}

\author[1]{Reijo Jaakkola}
\author[1]{Tomi Janhunen}
\author[1,2]{Antti Kuusisto}
\author[1]{Masood~Feyzbakhsh~Rankooh}
\author[1,\footnote{The 
authors are given in the alphabetical order.}]{Miikka Vilander}
\affil{Tampere University, Finland}
\affil[2]{University of Helsinki, Finland}

%
%
%
%

%
\maketitle 

\begin{abstract}
\noindent
We investigate explainability via short Boolean formulas in the data model based on unary relations. As an explanation of length $k$, we take a Boolean formula of length $k$ that minimizes the error with respect to the target attribute to be explained. We first provide novel quantitative bounds for the expected error in this scenario. We then also demonstrate how the setting works in practice by studying three concrete data sets. In each case, we calculate explanation formulas of different lengths using an encoding in Answer Set Programming. The most accurate formulas we obtain achieve errors similar to other methods on the same data sets. However, due to overfitting, these formulas are not necessarily ideal explanations, so we use cross validation to identify a suitable length for explanations. By limiting to shorter formulas, we obtain explanations that avoid overfitting but are still reasonably accurate and also, importantly, human interpretable. 
\end{abstract}


\section{Introduction}

In this article we investigate explainability and classification via short Boolean formulas. As the data model, we use multisets of propositional assignments. This is one of the simplest data representations available---consisting simply of data points and properties---and corresponds precisely to relational models with unary relations. The data is given as a model $M$ with unary relations $p_1,\dots , p_k$ over its domain $W$, and furthermore, there is an additional \emph{target predicate} $q\subseteq W$. As classifiers for recognizing $q$, we produce Boolean formulas $\varphi$ over $p_1,\dots , p_k$, and the corresponding error is then the percentage of points in $W$ that disagree on $\varphi$ and $q$ over $W$. For each formula length $\ell$, a formula producing the minimum error is chosen as a candidate classifier. Longer formulas produce smaller errors, and ultimately the process is halted based on cross validation which shows that the classifier formulas $\varphi$ begin performing significantly better on training data in comparison to test data, suggesting overfitting.

Importantly, the final classifier formulas $\varphi$ tend to be short and therefore \emph{explicate} the global behaviour of the classifier $\varphi$ itself in a transparent way. This leads to \emph{inherent interpretability} of our approach. Furthermore, the formulas $\varphi$ can \emph{also} be viewed as \emph{explanations} of the target predicate $q$. By limiting to short formulas, we obtain explanations (or classifiers) that avoid overfitting but are still reasonably accurate and also---importantly---human interpretable. 
It is also worth mentioning that, in addition to target attributes $q$ that occur in data, we could also explain a target $q$ that arises from the decisions of a (possibly black box) classifier. Altogether, the short Boolean formulas $\varphi$ we find can be understood in at least three different ways:
\begin{itemize}
    \item classifiers for target predicates $q$ in data,
    \item explanations of target predicates $q$ in data, 
    \item explanations of external classifiers, the decisions of which are given by $q$.
\end{itemize}

The last case involves a range of different approaches---used for different purposes---for searching
the formulas: see Section \ref{sec:explaining-a-classifier} for further details. 




Our contributions include theory, implementation and empirical results. We begin with some theory on the errors of Boolean formulas as explanations.
We first investigate general reasons behind overfitting when using Boolean formulas. We also observe, for example, that if all distributions are equally likely, the expected \emph{ideal theoretical error} of a distribution is $25\%$. The ideal theoretical error is the error of an ideal Boolean classifier for the entire distribution. We proceed by proving novel, quantitative upper and lower bounds on the expected \emph{ideal empirical error} on a data set sampled from a distribution. The ideal empirical error is the smallest error achievable on the data set. Our bounds give concrete information on sample sizes required to avoid overfitting.

We also compute explanation formulas in practice. We use three data sets from the UCI machine learning repository: Statlog (German Credit Data), Breast Cancer Wisconsin (Original) and Ionosphere.  
We obtain results comparable to 
other experiments in the literature. 
In one set of our experiments, the 
empirical errors for the obtained classifiers for the credit, breast cancer and ionosphere data are 0.27, 0.047 and 0.14. The corresponding formulas are surprisingly short, with lengths 6, 8 and 7, respectively. This makes them highly interpretable. The length 6 formula
for the credit data (predicting if a
loan will be granted) 
is 

\medskip

{\centering
$\displaystyle \neg(a[1,1] \wedge a[2]) \vee a[17,4],$\par}

\medskip

\noindent
where $a[1,1]$ means negative account balance; 
$a[2]$ means above median loan duration; 
and $a[17,4]$ means employment on managerial level.
Our errors are comparable to those obtained for the same data sets in the literature. For example, \cite{YangW01}
obtains an error 0.25 
for the credit data where our 
error is 0.27. Also, all our formulas are 
immediately interpretable. 
See Section \ref{section:results} for further discussion. 


On the computational side, we deploy answer set programming (ASP; see, e.g.,
\cite{BET11:jacm,JN16:aimag}) where the solutions of a search problem are described
declaratively in terms of rules such that the \emph{answer sets} of the resulting logic program correspond to the solutions of the problem. Consequently, dedicated search engines, known as
\emph{answer-set solvers},
provide means to solve the problem via the computation of answer
sets. The \system{Clasp} \cite{GKK0S15:lpnmr} and \system{Wasp}
\cite{ADLR15:lpnmr} solvers represent the state-of-the art of native
answer set solvers, providing a comparable performance in practice.
These solvers offer various reasoning modes---including prioritized
optimization---which are deployed in the sequel, e.g., for the
minimization of error and formula length. Besides these features, we
count on the flexibility of rules offered by ASP when describing
explanation tasks. More information on the technical side of ASP
can be found from the de-facto reference manual \cite{GKKS12:book}
of the \system{Clingo} system.

The efficiency of explanation is governed by the number of hypotheses considered basically in two ways. Firstly, the search for a plausible explanation requires the exploration of the hypothesis space and, secondly, the exclusion of better explanations becomes a further computational burden, e.g., when the error with respect to data is being minimized.
In computational learning approaches (cf.~\cite{Mitchell82:aij}), such as \emph{current-best-hypothesis search} and \emph{version space learning}, a hypothesis in a normal form is maintained while minimizing the numbers of false positive/negative examples. However, in this work, we tackle the hypothesis space somewhat differently: we rather specify the form of hypotheses and delegate their exploration to an (optimizing) logic solver. In favor of interpretability, we consider formulas based on negations, conjunctions, and disjunctions, not necessarily in a particular normal form. By changing the form of hypotheses, also other kinds of explanations such as decision trees \cite{Quinlan86ml} or lists could alternatively be sought.

Concerning further related work, our bounds on the difference between theoretical and empirical error are technically related to results in PAC learning \cite{IntroCompLearningTheory,Valiant84:cacm}. In PAC learning, the goal is to use a sample of labeled points drawn from an unknown distribution to find a hypothesis that gives a small true error with high probability.
The use of hypotheses that have small descriptions has also been considered in the PAC learning literature in relation to the principle of Occam's razor \cite{BEHW87:ipl,BEHW89:jacm,BP92:tcs}. One major difference between our setting and PAC learning is that in the latter, the target concept is a (usually Boolean) function of the attribute values, while in our setting we only assume that there exists a probability distribution on the propositional types over the attributes.

Explanations relating to minimality notions in relation to different 
Boolean classifiers have been studied widely, see for example \cite{ShihCD18} for \emph{minimum-cardinality} and \emph{prime implicant} explanations, also in line with Occam's razor \cite{BEHW87:ipl}.
Our study relates especially to global (or general \cite{shortfor}) explainability, where the full behaviour of a classifier is explained instead of 
a decision concerning a particular input instance. Boolean complexity---the length of the shortest equivalent formula{\hspace{0.01pt}}---is promoted in the prominent article \cite{feldman} as an empirically tested measure of the subjective difficulty of a concept. On a 
conceptually 
related note, intelligibility of various Boolean classifiers are studied in \cite{audemard}. While that study places,  e.g., DNF-formulas to the less intelligible category based on the complexity of explainability queries  performed on classifiers, we note that with genuinely small bounds for classifier length, asymptotic complexity can sometimes be a somewhat problematic measure for intelligibility. In our study, the bounds arise already from the overfitting thresholds in real-life data. 
In the scenarios we studied, overfitting indeed sets natural, small bounds for classifier length. In inherently Boolean data, such bounds can be fundamental and cannot be always ignored via using different classes of classifiers. 
The good news 
is that while a length bound may be \emph{necessary} to avoid overfitting, \emph{shorter length increases interpretability}. This is important from the point of theory as well as applications.

The article is organized as follows.
After the preliminaries in Section \ref{section:preliminaries},
we present theoretical results on errors in
Section \ref{section:expected-errors}.
Section \ref{section:implementation} explains the 
ASP implementation.
In Section \ref{section:results} we present and interpret the empirical results.

\section{Preliminaries}
\label{section:preliminaries}

The syntax of propositional logic $\PL[\sigma]$ over the vocabulary $\sigma = \{p_1, \dots, p_m\}$ is given by
$\varphi ::= p \mid \neg \varphi \mid \varphi \land \varphi \mid \varphi \lor \varphi$
where $p \in \sigma$.  We also define the exclusive or $\varphi \oplus \psi := (\varphi \lor \psi) \land \neg(\varphi \land \psi)$ as an abbreviation. A \textbf{$\sigma$-model} is a
structure $M = (W,V)$ where $W$ is a finite, non-empty set referred to as
the \textbf{domain} of $M$ and $V:\sigma \rightarrow \mathcal{P}(W)$ is a
\textbf{valuation} function that assigns each $p\in \sigma$ the set $V(p)$ (also denoted by $p^M$) of
points $w\in W$ where $p$ is
considered to be true.  

A $\sigma$-valuation $V$ can be extended in the standard way to a valuation $V: \PL[\sigma] \to \mathcal{P}(W)$ for all $\PL[\sigma]$-formulas. We write $w\models \varphi$ if 
$w\in V(\varphi)$ and say that $w$ \textbf{satisfies} $\varphi$. We denote by $|\varphi|_M$ the \emph{number} of points $w \in W$ where $\varphi \in \PL[\sigma]$ is true. For $\sigma$-formulas $\varphi$ and $\psi$, we write $\varphi\models\psi$ iff for all $\sigma$-models $M = (W, V)$ we have $V(\psi) \subseteq V(\varphi)$.
Let $\mathit{lit}(\sigma)$ denote the set of \textbf{$\sigma$-literals}, i.e., formulas $p$ and $\neg p$ for $p\in \sigma$.
A \textbf{$\sigma$-type base} is a set $S\subseteq \mathit{lit}(\sigma)$ 
such that for each $p\in \sigma$, precisely one 
of the literals $p$ and $\neg p$ is in $S$. A \textbf{$\sigma$-type} is a
conjunction $\bigwedge S$. We assume some fixed bracketing and ordering of literals in $\bigwedge S$ so there is a one-to-one correspondence between type bases and types. The set of $\sigma$-types is denoted by $T_{\sigma}$. Note that in a $\sigma$-model $M = (W,V)$, each element $w$ satisfies precisely one $\sigma$-type, so the domain $W$ is partitioned by some subset of $T_{\sigma}$. The \textbf{size} $\size(\varphi)$ of a formula $\varphi \in \PL[\sigma]$ is defined such that $\size(p) = 1,$ 
$\size(\neg \psi) = \size(\psi) + 1,$
and $\size(\psi \land \vartheta) = \size(\psi \lor \vartheta) = \size(\psi) + \size(\vartheta) + 1$. 
%
%
%

%
%
%

We will use short propositional formulas as \emph{explanations} of target attributes in data.
Throughout the paper, we shall use the vocabulary $\tau = \{p_1,\dots , p_k\}$ for the language of explanations, while $q\not \in \tau$ will be the target 
attribute (or target proposition) to be explained. While the set of $\tau$-types will be denoted by $T_{\tau}$, we let $T_{\tau,q}$ denote the set of $(\tau\cup \{q\})$-types in the extended language $\PL[\tau\cup\{q\}]$.

By a probability distribution over a 
vocabulary $\sigma$, or 
simply a $\sigma$-distribution, we mean a
function $\mu_{\sigma} : T_{\sigma} \rightarrow [0,1]$ that gives a probability to each type in $T_{\sigma}$. We are mainly interested in such distributions over $\tau$ and $\tau\cup \{q\}$. For notational convenience, we may write $\mu_{\tau,q}$ or simply $\mu$ instead of $\mu_{\tau\cup\{q\}}$. 
In the theoretical part of the paper, we assume that the studied data (i.e., $(\tau\cup \{q\})$-models) are sampled using such a distribution $\mu$.  

We then define some notions of error for explanations. 
Let $\tau = \{p_1, \dots, p_k\}$. 
Fix a probability distribution $\mu:T_{\tau, q} \rightarrow [0,1]$. 
Let $\varphi$
and $\psi$ be $(\tau \cup \{q\})$-formulas.
The \textbf{probability of $\varphi$ over $\mu$} is
defined as

\medskip

{\centering
$
\displaystyle \Pr\nolimits_{\mu}(\varphi) := 
\sum\limits_{t\in T_{\tau,q}\hspace{0.1mm},\ t\models\varphi} \mu(t). 
$
\par}

\smallskip

%
%
%
\noindent The \textbf{probability of $\psi$
given $\varphi$ over $\mu$} is
defined as $\Pr_{\mu}(\psi\, |
\, \varphi) := \frac{\Pr_\mu(\psi\wedge\varphi)}{\Pr_{\mu}(\varphi)}$
(and $0$ if $\Pr_\mu(\varphi) = 0$).
For simplicity, we may write $\mu(\varphi)$ for $\Pr_{\mu}(\varphi)$
and $\mu(\psi\, |\, \varphi)$
for $\Pr_{\mu}(\psi\, |\, \varphi)$.
Let $M = (W,V)$ be a $(\tau \cup \{q\})$-model. 
The \textbf{probability of $\varphi$
over $M$} is
$\Pr_M(\varphi):=\frac{1}{|W|}|\varphi|_M$, and 
the \textbf{probability of $\psi$ given $\varphi$ over $M$} is defined as $\Pr_M(\psi\, |\, \varphi) := \frac{|\psi\wedge\varphi|_M}{|\varphi|_M}$
(and $0$ if $\Pr_M(\varphi) = 0$).
The disjunction $\varphi_{id}^M
:= \bigvee \{\, t\in T_{\tau}\, |\, 
 \Pr_M(q\, |\, t) \geq \frac{1}{2} \}$ is
 the \textbf{ideal 
classifier} w.r.t. $M$, 
and the disjunction $\varphi_{id}^{\mu}
:= \bigvee \{\, t\in T_{\tau}\, |\, 
 \mu(q\, |\, t) \geq \frac{1}{2} \}$ is
 the \textbf{ideal 
classifier} w.r.t. $\mu$.

Now, let $\psi \in \PL[\tau]$. The \textbf{theoretical error} (or \textbf{true error}) of $\psi$ with respect to $\mu$ is
$
\err_\mu(\psi) := \Pr_\mu(\psi \oplus q).
$
The \textbf{ideal theoretical error of $\mu$} is
\begin{align*}
  \err(\mu) := \min_{\psi \in \PL[\tau]} \err_\mu(\psi)
= \err_{\mu}(\varphi_{id}^{\mu})
&= \Pr_\mu(\varphi_{id}^{\mu} \oplus q) \\
&=
\sum_{t \in T_\tau} \min\{\mu(t \land q), \mu(t \land \neg q)\}.  
\end{align*}

\medskip

\noindent 
Let $M$ be a $(\tau \cup \{q\})$-model. The \textbf{empirical error} of $\psi$ with respect to $M$ is $\err_M(\psi) := \Pr_M(\psi \oplus q).$
%
%
%
%
%
%
%
%
%
%
%
%
%
%
%
The \textbf{ideal empirical error} of $M$ is
\begin{align*}
    \mathrm{err}(M) := \min_{\psi \in \PL[\tau]} \mathrm{err}_M(\psi)
= \err_M(\varphi_{id}^M)
&= \Pr_M(\varphi_{id}^{\mu} \oplus q) \\
&=
\frac{1}{|W|} \sum_{t\, \in\, T_\tau} \min \{|t \land q|_M, |t \land \neg q|_M\}.
\end{align*}

\medskip

\noindent
For a $\tau$-type $t$,  
the \textbf{ideal
error over $t$ w.r.t. $\mu$} is $\min\{\mu(q\, |\, t), 
\mu(\neg q\, |\, t)\}$. 
The \textbf{ideal error over $t$ w.r.t. $M$} is 
$\min\{\Pr_M(q\, |\, t), \Pr_M(\neg q\, |\, t)\}$.

The main problem studied in this paper is the following: over a
$(\tau\cup\{q\})$-model $M$, given a
bound $\ell$ on formula length, find $\psi$
with $\size(\psi) \leq \ell$ and with 
minimal empirical error w.r.t. $M$. 
This can be formulated as a \emph{general explanation problem} (GEP) in
the sense of \cite{shortfor}; see in particular 
the extended problems in \cite{shortforarxiv}. The goal in GEP is to 
explain the \emph{global} behaviour of a classifier rather than a 
reason why a particular instance was accepted or rejected.

Finally, we define $\mathit{cut}: [0,1]\rightarrow [0,\frac{1}{2}]$ to be
the function such that $\mathit{cut}(x) = x$ if $x\leq \frac{1}{2}$ and otherwise $\mathit{cut}(x) = 1 - x$.

\section{Expected errors}
\label{section:expected-errors}

In this section we consider the errors given by 
Boolean classifiers, including the phenomena
that give rise to the errors. 
With no information on 
the distribution $\mu : T_{\tau, q} \rightarrow [0,1]$, it is
difficult to predict the error of a 
classifier $\varphi$ in $PL[\tau]$. 
However, some observations can be made. Consider the scenario where all distributions $\mu$
are equally likely, meaning that we consider the 
flat Dirichlet distribution $Dir(\alpha_1,\dots , \alpha_{|T_{\tau, q}|})$
with each $\alpha_i$ equal to $1$,
i.e., the distribution that is uniform over its support which, in
turn, is 
the $(|T_{\tau, q}|-1)$-simplex. For 
more 
on Dirichlet distributions, see \cite{kotz2004continuous}. We begin with the 
following 
observation.
%
%
%
%
%
%
%

\begin{proposition}\label{errorproposition}
Assuming all distributions over $\tau\cup\{q\}$ are equally likely, the expected value of the ideal theoretical error is $0.25$. Also, for 
any type $t\in T_{\tau}$ and any $\mu_{\tau}$
with $\mu_{\tau}(t) > 0$, if 
all extensions $\mu$ of $\mu_{\tau}$
to a $(\tau\cup \{q\})$-distribution 
are equally likely, 
the expectation of the ideal error over $t$
w.r.t. $\mu$ is likewise $0.25$.  
\end{proposition}

\begin{proof}
We prove the second claim first. 
Fix a $\mu$ and $t$.
If $x = \mu(q\, |\, t)$, then the ideal error over $t$ w.r.t. $\mu$ is given by $\mathit{cut}(x)$. 
%
%
%
%
%
%
Therefore the corresponding 
expected value is given by 
\[\frac{1}{1-0}\int_0^1 \mathit{cut}(x)\, dx\ = \ \int_{0}^{\frac{1}{2}} x\, dx\ +\
\int_{\frac{1}{2}}^{1} (1 - x)\, dx\ =\ \frac{1}{4}.\] 
%
%
%
%
%
%
This proves the second claim. 
Based on this, it is not difficult to show
that the also the first claim holds;
the full 
details are given in the Appendix.
\end{proof}


One of the main problems with Boolean 
classifiers is that the number of
types is exponential in the vocabulary
size, i.e., the curse of dimensionality. This leads to
overfitting via overparametrisation; even if
the model $M$ is faithful to an 
underlying distribution $\mu$, 
classifiers $\varphi_{id}^M$ tend to give 
empirical errors that are significantly
smaller than the
theoretical ones for $\mu$. To see why,
notice that in the 
extreme case where $|t|_M = 1$ for each $t\in T_{\tau, q}$, the ideal
empirical error of $M$ is zero. In general, when the sets $|t|_M$ are small, ideal classifiers $\varphi_{id}^M$ benefit from that. Let us consider this
issue quantitatively. 
Fix $\mu$ and $t\in T_{\tau}$.
For a model $M$, let  $\err(M,t)$ refer to
the ideal error over $t$ w.r.t. $M$. 
Consider models $M$ sampled 
according to $\mu$, and let $m\in \mathbb{N}$
and $\mu(q\, |\, t) = p$.  
Now, the expected value $E(m,p)$ of $\err(M,t)$
over those models $M$ where $|t|_M = m$ is given by
\medskip

\smallskip 

{\centering
$
\displaystyle \Bigl(\sum\limits_{0\, <\, k\, \leq \, m/2 }
\binom{m}{k}p^k(1-p)^{m-k} \cdot \frac{k}{m}\Bigr) 
+ \Bigl(\sum\limits_{m/2\, <\, k\, < \, m }
\binom{m}{k}p^k(1-p)^{m-k} \cdot \frac{(m - k)}{m}\Bigr).
$
\par}

\medskip 

\smallskip

\noindent
Now for example $E(4,0.7) = 0.2541$
and $E(2,0.7) = 0.21$, both significantly lower
than $\mathit{cut}(p) = \mathit{cut}(0.7) = 0.3$ which is the 
expected value of $\err(M,t)$ when 
the size restriction $|t|_M = m$ is lifted and models of increasing size are sampled according to $\mu$.
Similarly, we 
have $E(4,0.5) = 0.3125$
and $E(2,0.5) = 0.25$, significantly 
lower than $\mathit{cut}(p) =
\mathit{cut}(0.5) = 0.5$. A natural way to
avoid this phenomenon is to limit formula size, the strategy adopted in this paper. This also naturally leads to
shorter and thus more interpretable formulas.


We next estimate empirical errors for general Boolean classifiers (as opposed to single types). The \textbf{expected ideal empirical error of} $\mu$ is simply the expectation $\mathbb{E}(\mathrm{err}(M))$ of $\mathrm{err}(M)$, where $M$ is a model of size $n$ sampled according to $\mu$. One can show that $\mathbb{E}(\mathrm{err}(M)) \leq \mathrm{err}(\mu)$ and that $\mathbb{E}(\mathrm{err}(M)) \to \mathrm{err}(\mu)$ as $n \to \infty$. Thus it is natural to ask how the size of the difference $\mathrm{err}(\mu) - \mathbb{E}(\mathrm{err}(M))$, which we call the \textbf{expected error gap}, depends on $n$.

In the remaining part of this section we establish bounds on the expected ideal empirical error, which in turn can be used to give bounds on the expected error gap. Since expectation is linear, it suffices to give bounds on 
\begin{equation}\label{eq:expected_empirical_error}
    \frac{1}{n}\sum_{t \in T_\tau} \mathbb{E}\min\{|t \land q|_M,|t \land \neg q|_M\},
\end{equation}
where $M$ is a model of size $n$ which is sampled according to $\mu$. Here, for each type $t\in T_\tau$, $|t \land q|_M$ and $|t \land \neg q|_M$ are random variables that are distributed according to $\mathrm{Binom}(n,\mu(t \land q))$ and $\mathrm{Binom}(n,\mu(t \land \neg q))$ respectively. Since $|t \land q|_M + |t \land \neg q|_M = |t|_M$, we can replace $|t \land \neg q|_M$ with $|t|_M - |t \land q|_M$.

To simplify (\ref{eq:expected_empirical_error}), we will first use the law of total expectation to write it as
\begin{equation}\label{eq:conditioned_expected_empirical_error}
    \frac{1}{n} \sum_{t \in T_\tau} \sum_{m = 0}^n \mathbb{E}(\min\{|t\land q|_M,m - |t\land q|_M\} \mid |t|_M = m) \cdot \Pr(|t|_M = m).
\end{equation}
For each $0 \leq m \leq n$ and $t \in T_\tau$ we fix a random variable $X_{m,t,q}$ distributed according to $\mathrm{Binom}(m,\mu(q|t))$, where $\mu(q|t) := \mu(t \land q)/\mu(t)$. In the Appendix we show that (\ref{eq:conditioned_expected_empirical_error}) equals
\begin{equation}\label{eq:simplified_expected_empirical_error}
    \frac{1}{n} \sum_{t \in T_\tau} \sum_{m = 0}^n \mathbb{E}\min\{X_{m,t,q}, \ m - X_{m,t,q}\} \cdot \Pr(|t|_M = m).
\end{equation}
To avoid dealing directly with the expectation of a minimum of two Binomial random variables, we will simplify it using the identity $\min\{a,b\} = \frac{1}{2}(a + b - |a - b|).$
In the Appendix we show that using this identity on (\ref{eq:simplified_expected_empirical_error}) gives the form
\begin{equation}\label{eq:alternative_expected_empirical_error_formula}
    \frac{1}{2} - \frac{1}{n} \sum_{t \in T_\tau} \sum_{m = 0}^n \mathbb{E}\bigg|X_{m,t,q} - \frac{m}{2}\bigg| \cdot \Pr(|t|_M = m).
\end{equation}
In the above formula the quantity $\mathbb{E}|X_{m,t,q} - \frac{m}{2}|$ is convenient 
since we can bound it from above using the standard deviation of $X_{m,t,q}$. Some further estimates and algebraic manipulations in the Appendix suffice to prove the following result.

\begin{theorem}\label{thm:expected_empirical_error_lower_bound}
    Expected ideal empirical error is bounded from below by
    \[\err(\mu) - \frac{1}{\sqrt{n}} \sum_{t \in T_\tau} \sqrt{\mu(q|t)(1 - \mu(q|t))\mu(t)}.\]
\end{theorem}

We note that Theorem \ref{thm:expected_empirical_error_lower_bound} implies immediately that the expected error gap is bounded from above by
$\frac{1}{\sqrt{n}} \sum_{t \in T_\tau} \sqrt{\mu(q|t)(1 - \mu(q|t))\mu(t)} \leq \frac{1}{2}\sqrt{\frac{|T_\tau|}{n}}.$
This estimate yields quite concrete sample bounds. For instance, if we are using three attributes to explain the target attribute (so $|T_\tau| = 8$) and we want the expected error gap to be at most $0.045$, then a sample of size at least $1000$ suffices. For the credit data set with $1000$ data points, this means that if three attributes are selected, then the 
(easily computable)
ideal empirical error
gives a good idea of the ideal theoretical error for those three attributes.

Obtaining an upper bound on the expected ideal empirical error is much more challenging, since in general it is not easy to give good lower bounds on $\mathbb{E}|X - \lambda|$, where $X$ is a binomial random variable and $\lambda > 0$ is a real number. Nevertheless we were able to obtain the following result.

\begin{theorem}\label{thm:upper_bound_expected_empirical_error}
    Expected ideal empirical error is bounded from above by
    \medskip
    
    {\centering
    $\displaystyle\frac{1}{2} - \frac{1}{\sqrt{8n}}\sum_{n\mu(t) \geq 1} \sqrt{\mu(t)} + \frac{1}{2\sqrt{8}n}\sum_{n\mu(t) \geq 1} \frac{1 - \mu(t)}{\sqrt{n\mu(t)}} - \frac{1}{\sqrt{8}} \sum_{n\mu(t) < 1} \mu(t)(1 - \mu(t))^n.$\par}
    
    \medskip
\end{theorem}

The proof of Theorem \ref{thm:upper_bound_expected_empirical_error} in the Appendix can be divided into three main steps. First, we observe that the expected ideal empirical error is maximized when $\mu(q|t) = 1/2$, for every $t \in T_\tau$, in which case $\mathbb{E}(X_{m,t,q}) = \frac{m}{2}$. Then, we use a very recent result of \cite{PELEKIS2016305} to obtain a good lower bound on the value $\mathbb{E}|X_{m,t,q} - \mathbb{E}(X_{m,t,q})|$. Finally, after some algebraic manipulations, we are left with the task of bounding $\mathbb{E}(\sqrt{|t|_M})$ from below, which we achieve by using an estimate that can be obtained from the Taylor expansion of $\sqrt{x}$ around $1$.

To get a concrete feel for the lower bound of Theorem \ref{thm:upper_bound_expected_empirical_error}, consider the case where $\mu(q|t) = 1/2$, for every $t \in T_\tau$. In this case a \emph{rough} use of Theorem \ref{thm:upper_bound_expected_empirical_error} implies that the expected error gap is bounded from below by

\medskip

{\centering
$\displaystyle \frac{1}{\sqrt{8}} \sum_{n\mu(t) < 1} \mu(t)(1 - \mu(t))^n \geq \frac{1}{\sqrt{8}e} \cdot \frac{(n-1)}{n} \cdot \sum_{n\mu(t) < 1} \mu(t).$\par}

\medskip

\noindent
This lower bound very much depends on the properties of the distribution $\mu$, but one can nevertheless make general remarks about it. For instance, if $|T_\tau|$ is much larger than $n$ and $\mu$ is not concentrated on a small number of types (i.e., its Shannon entropy is not small), then we except $\sum_{n\mu(t) < 1} \mu(t)$ to be close to one. Thus the above bound would imply that in this scenario the generalization gap is roughly $1/(\sqrt{8} \cdot e) \approx 0.13$, which is a significant deviation from zero.

\section{An Overview of the Implementation in ASP}
\label{section:implementation}

In this section, we describe our proof-of-concept implementation of the search for short formulas explaining data sets. Our implementation presumes Boolean attributes only and \emph{complete} data sets having no missing values. 
%
%
In the following, we highlight the main ideas behind our ASP encoding in terms of code snippets in the Gringo syntax \cite{GKKS12:book}.
The complete encoding will be published under the ASPTOOLS collection%
\footnote{\url{https://github.com/asptools/benchmarks}}
along with some preformatted data sets for testing purposes.
Each data set is represented in terms of a predicate \code{val(D,A,V)} with three arguments: \code{D} for a data point identifier, \code{A} for the name of an attribute, and \code{V} for the value of the attribute \code{A} at \code{D}, i.e., either \code{0} or \code{1} for Boolean data.
%

\begin{lstlisting}[label=code:domains,frame=single,float=t,escapechar=?,caption={Encoding the syntactic structure of hypotheses}]
% Domains
#const l=10.                            ?\label{line:length}?
node(1..l).  root(l).  op(neg;and;or).  ?\label{line:domains}?
data(D) :- val(D,A,B).                  ?\label{line:data}?
attr(A) :- val(D,A,B).                  ?\label{line:attr}?

% Choose the actual length
{used(N)} :- node(N).                   ?\label{line:used}?
used(N+1) :- used(N), node(N+1).        ?\label{line:used-next}?
used(N) :- root(N).                     ?\label{line:used-root}?

%  Choose leaf nodes and inner nodes, and label them
{leaf(N)} :- used(N).                   ?\label{line:leaf}?
inner(N) :- used(N), not leaf(N).       ?\label{line:inner}?
{ op(N,O): op(O) } = 1 :- inner(N).     ?\label{line:operator}?
{ lat(N,A): attr(A) } = 1 :- leaf(N).   ?\label{line:leaf-attr}? 
\end{lstlisting}

Given a data set based on attributes $\rg{a_0}{,}{a_n}$ where $a_n$ is the target of explanation, the hypothesis space is essentially the propositional language $\PL[\tau]$ with the vocabulary $\tau=\eset{a_0}{a_{n-1}}$. Thus, the goal is to find a definition $a_n\lequiv\varphi$ where $\varphi\in\PL[\tau]$ with the least error. To avoid obviously redundant hypotheses, we use only propositional connectives from the set $C=\set{\neg,\land,\lor}$ and represent formulas in the so-called \emph{reversed Polish notation}. This notation omits parentheses altogether and each formula $\varphi$ is encoded as a sequence $\rg{s_1}{,}{s_k}$ of symbols where $s_i\in\tau\union C$ for each $s_i$. Such a sequence can be transformed into a formula by processing the symbols in the given order and by pushing formulas on a \emph{stack} that is empty initially. If $s_i\in\tau$, it is pushed on the stack, and if $s_i\in C$, the arguments of $s_i$ are popped from the stack and the resulting formula is pushed on the stack using $s_i$ as the connective. Eventually, the result appears as the only formula on top of stack. For illustration, consider the sequence $a_2,a_1,\land,\neg,a_0,\lor$ referring to attributes $a_0$, $a_1$, and $a_2$. The stack evolves as follows:
$a_2$ $\mapsto$
$a_2,a_1$ $\mapsto$
$(a_1\land a_2)$ $\mapsto$
$\neg(a_1\land a_2)$ $\mapsto$
$\neg(a_1\land a_2),a_0$ $\mapsto$
$a_0\lor\neg(a_1\land a_2)$.
Thus, the formula is $a_0\lor\neg(a_1\land a_2)$.
For a formula $\varphi$, the respective sequence of symbols can be found by traversing the syntax tree of $\varphi$ in the \emph{post order}. There are also malformed sequences not corresponding to any formula.

\begin{lstlisting}[label=code:stack,frame=single,float=t,escapechar=?,caption={Checking the syntax using a stack}]
% Check the size of the stack
count(N,0) :- used(N), not used(N-1).                       ?\label{line:count-base}?
count(N+1,K+1) :- leaf(N), count(N,K), node(N), K>=0, K<=2. ?\label{line:count-leaf}?
count(N+1,K) :- count(N,K), node(N), op(N,neg).             ?\label{line:count-neg}?
count(N+1,K-1) :- count(N,K), node(N), op(N,O), O!=neg.     ?\label{line:count-other}?
:- not count(l+1,1).                                        ?\label{line:count-final}?

%  The step-by-step evolution of the stack
stack(N+1,K+1,N) :- leaf(N), count(N,K), K>=0, K<=2.         ?\label{line:stack-leaf}?
stack(N+1,K,  N) :- op(N,neg), count(N,K), K>0, K<=3.        ?\label{line:stack-unary}?
stack(N+1,K-1,N) :- op(N,O), O!=neg, count(N,K), K>=2, K<=3. ?\label{line:stack-binary}?

stack(N+1,I,  M) :- leaf(N), count(N,K), I>=0, I<=K, stack(N,I,M). ?\label{line:copy-leaf}?
stack(N+1,I,  M) :- op(N,neg), count(N,K), I>0, I<K, stack(N,I,M). ?\label{line:copy-unary}? 
stack(N+1,1,  M) :- op(N,O), O!=neg, count(N,3), stack(N,1,M).     ?\label{line:copy-binary}?
\end{lstlisting}

Based on the reverse Polish representation, the first part of our encoding concentrates on the generation of hypotheses whose syntactic elements are defined in Listing \ref{code:domains}. In Line \ref{line:length}, the maximum length of the formula is set, as a global parameter \code{l} of the encoding, to a default value \code{10}. Other values can be issued by the command-line option \code{-cl=<number>}. Based on the value chosen, the respective number of \emph{nodes} for a syntax tree is defined in Line \ref{line:domains}, out of which the last one is dedicated for the \emph{root}. The three Boolean \emph{operators} are introduced using the predicate \code{op/1}. The data points and attributes are extracted from data in Lines \ref{line:data} and \ref{line:attr}, respectively. To allow explanations shorter than \code{l}, the choice rule in Line \ref{line:used} may take any node into use (or not). The rule in Line \ref{line:used-next} ensures that all nodes with higher index values up to \code{l} are in use. The root node is always in use by Line \ref{line:used-root}. The net effect is that the nodes \code{i..l} taken into use determine the actual length of the hypothesis. Thus the length may vary between \code{1} and \code{l}.
In a valid syntax tree, the nodes are either \emph{leaf} or \emph{inner} nodes, see Lines \ref{line:leaf} and \ref{line:inner}, respectively. Each inner node is assigned an operator in Line \ref{line:operator} whereas each leaf node is assigned an attribute in Line \ref{line:leaf-attr}, to be justified later on.

The second part of our encoding checks the syntax of the hypothesis using a stack, see Listing \ref{code:stack}. Line \ref{line:count-base} resets the size of the stack in the first used node. The following rules in Lines \ref{line:count-leaf}--\ref{line:count-other} take the respective effects of attributes, unary operators, and binary operators into account. The constraint in Line \ref{line:count-final} ensures that the count reaches \code{1} after the root node. Similar constraints limit the size of the stack: at most two for leaf nodes and at least one/two for inner nodes with a unary/binary connective.
The predicate \code{stack/3} propagates information about arguments to operators, i.e., the locations \code{N} where they can be found. Depending on node type, the rules in Lines \ref{line:stack-leaf}--\ref{line:stack-binary} create a new reference that appears on top of the stack at the next step \code{N+1} (cf. the second argument \code{K+1}, \code{K}, or \code{K-1}). The rules in Lines \ref{line:copy-leaf}--\ref{line:copy-binary} copy the items under the top item to the next step \code{N+1}.

The third part of our encoding evaluates the chosen hypothesis at data points \code{D} present in the data set given as input. For a leaf node \code{N}, the value is simply set based on the value of the chosen attribute \code{A} at \code{D}, see Line \ref{line:true-leaf}. For inner nodes \code{N}, we indicate a choice of the truth value in Line \ref{line:true-inner}, but the choice is made deterministic in practice by the constraints in Lines \ref{line:eval-or1}--\ref{line:eval-or3}, illustrating the case of the \code{or} operator. The constraints for the operators \code{neg} and \code{and} are analogous.

\begin{lstlisting}[label=code:evaluate,frame=single,float=t,escapechar=?,caption={Evaluating the hypothesis at data points}]
true(D,N) :- data(D), leaf(N), lat(N,A), val(D,A,1). ?\label{line:true-leaf}?
{true(D,N)} :- data(D), used(N), inner(N).           ?\label{line:true-inner}?

% Constraints for disjunctions
:- data(D), op(N,or), count(N,I), stack(N,I-1,N3),   ?\label{line:eval-or1}?
   true(D,N), not true(D,N-1), not true(D,N3).
:- data(D), op(N,or), not true(D,N), true(D,N-1).    ?\label{line:eval-or2}?
:- data(D), op(N,or), count(N,I), stack(N,I-1,N2),   
   not true(D,N), true(D,N2).                        ?\label{line:eval-or3}?
\end{lstlisting}

\begin{lstlisting}[label=code:objective,frame=single,float=t,escapechar=?,caption={Encoding the objective function}]
% Compute error
error(D) :- data(D), val(D,A,0), expl(A), true(D,N), root(N).     ?\label{line:error-false-positive}?
error(D) :- data(D), val(D,A,1), expl(A), not true(D,N), root(N). ?\label{line:error-false-negative}?

#minimize { 1@1,D: error(D);  1@0,N: used(N), node(N) }.          ?\label{line:optimality}?
\end{lstlisting}

Finally, Listing \ref{code:objective} encodes the objective function. Lines \ref{line:error-false-positive} and \ref{line:error-false-negative} spot data points \code{D} that are incorrect with respect to the attribute \code{A} being explained and the selected hypothesis rooted at \code{N}. For a false positive \code{D}, the hypothesis is true at \code{D} while the value of \code{A} is \code{0}. In the opposite case, the hypothesis is false while the value of \code{A} at \code{D} is \code{1}.
The criteria for minimization are given in Line \ref{line:optimality}. The number of errors is the first priority (at level \code{1}) whereas the length of the hypothesis is the secondary objective (at level \code{0}). Also, recall that the maximum length has been set as a parameter earlier. The optimization proceeds \emph{lexicographically} as follows: a formula that minimizes the number of errors is sought first and, once such an explanation has been found, the length of the formula is minimized additionally. So, it is not that crucial to set the (maximum) length parameter \code{l} to a particular value: the smaller values are feasible, too, based on the nodes in use.
The performance of our basic encoding can be improved by adding constraints to prune redundant hypotheses, sub-optimal answer sets, and candidates.


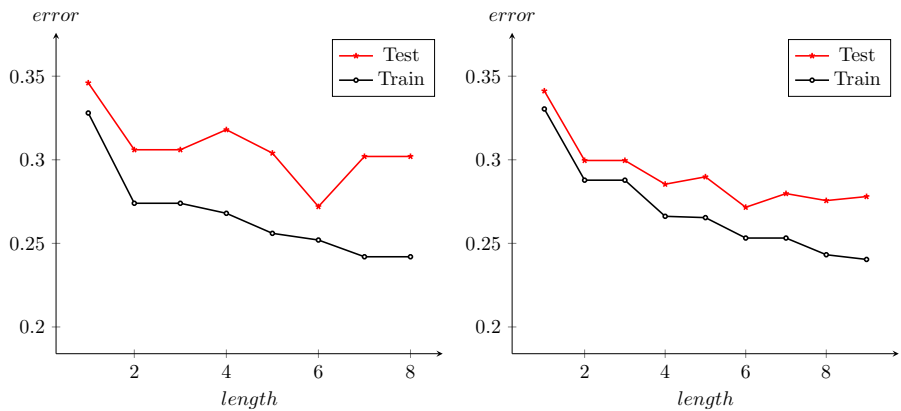
\begin{figure}[thb]
\centering
\begin{minipage}{.5\linewidth}
  \centering
\begin{tikzpicture}[scale=0.75]
\begin{axis}[  
         ymin=100/500, ymax=180/500,    xmin=1, xmax=8,    xlabel={$length$},    ylabel={$error$}, ylabel style={rotate=-90, at={(axis description cs:.08,1.05)}},  axis lines=left,    enlargelimits=true]
\addplot[    color=red,    thick,    mark=star, mark options={fill=white, scale=0.75}] coordinates {
    (1,173/500) (2,153/500) (3,153/500) (4,159/500) (5,152/500) (6,136/500) (7,151/500) (8,151/500)
};
\addplot[    color=black,    thick,    mark=*, mark options={fill=white, scale=0.5}] coordinates {
    (1,164/500) (2,137/500) (3,137/500) (4,134/500) (5,128/500) (6,126/500) (7,121/500) (8,121/500)
};
\legend{Test,Train}
\end{axis}
\end{tikzpicture}
\end{minipage}%
\begin{minipage}{.5\linewidth}
  \centering
  \begin{tikzpicture}[scale=0.75]
\begin{axis}[ 
ymin=100/500, ymax=180/500,    xmin=1, xmax=9,    xlabel={$length$},    ylabel={$error$}, ylabel style={rotate=-90, at={(axis description cs:.08,1.05)}},   axis lines=left,    enlargelimits=true]
\addplot[    color=red,    thick, mark=star, mark options={fill=white, scale=0.75}] coordinates {
    (1,170.6/500) (2,149.8/500) (3,149.8/500) (4,142.7/500) (5,144.9/500) (6,135.8/500) (7,139.9/500) (8,137.8/500) (9,139.00/500) 
};
\addplot[    color=black,    thick, mark=*, mark options={fill=white, scale=0.5}] coordinates {
    (1,165.20/500) (2,143.90/500) (3,143.90/500) (4,133.10/500) (5,132.70/500) (6,126.60/500) (7,126.60/500) (8,121.60/500) (9,120.20/500) 
};
\legend{Test, Train}
\end{axis}
\end{tikzpicture}
\end{minipage}
\caption{Credit data set -- first test (left) and average (right)\label{fig:credit-data}}
\end{figure}

\section{Results from data and interpretation}
\label{section:results}

To empirically analyze short Boolean formulas as explanations and classifiers, we utilize three data sets from the UCI machine learning repository: Statlog (German Credit Data), Breast Cancer Wisconsin (Original) and Ionosphere. The target attributes are given as acceptance of a loan application, benignity of a tumor and ``good'' radar readings, respectively.
The breast cancer data contains a small number of instances with missing attribute values (16 out of 699), which are excluded from the analysis.
%
The original data sets contain categorical and numerical attributes, as well as Boolean ones.
To convert a categorical attribute into Boolean format, we treat the inclusion of instances in each corresponding category as a separate Boolean attribute. For numerical attributes, we use the median 
across all instances as the threshold. Thus the Booleanized credit, breast cancer and ionosphere data sets consist of 1000, 683 and 351 instances, respectively, with 68, 9 and 34 Boolean attributes each, plus the target attribute. 
To evaluate the obtained formulas as classifiers, we randomly divide each data set into two equal parts: one serving as the training data and the other as the test data. For the training data $M$, target predicate $q$ and increasing formula length bounds $\ell$, we produce formulas $\psi$ not involving $q$ with $\size(\psi) \leq \ell$ that minimize the empirical error $\err_M(\psi)$. We also record the error on the test data (i.e., the complement of the training data).
We repeated this process 10 times.
For each data set, Figures \ref{fig:credit-data}, \ref{fig:cancer-data} and \ref{fig:ionosphere-data} record both the first experiment as an example and the average over 10 experiments on separately randomized training and test data sets. We employed \system{Clingo} (v.~5.4.0) as the answer-set solver in all experiments.

For the ionosphere data, the Booleanization via median is rough for the real-valued radar readings. Thus we expect larger errors compared to methods using real numbers. This indeed happens, but the errors are still surprisingly low.

\begin{figure}[!htb]
\centering
\begin{minipage}{.5\linewidth}
  \centering
\begin{tikzpicture}[scale=0.75]
\begin{axis}[    ymin=7/341, ymax=40/341,    xmin=1, xmax=18,    xlabel={$length$},    ylabel={$error$}, ylabel style={rotate=-90, at={(axis description cs:.08,1.05)}},   axis lines=left,    enlargelimits=true, xtick={1,3,5,7,9,11,13,15,17}, yticklabel style={/pgf/number format/.cd, fixed, fixed zerofill, precision=2}]

\addplot[    color=red,    thick,    mark=star, mark options={fill=white, scale=0.75}] coordinates {
    (1,198/341) (2,37/341) (3,37/341) (4,23/341) (5,23/341) (6,18/341) (7,18/341) (8,16/341) (9,16/341) (10,20/341) (11,20/341) (12,19/341) (13,19/341) (14,20/341) (15,20/341) (16,20/341) (17,20/341) (18,20/341)
};
\addplot[    color=black,    thick,    mark=*, mark options={fill=white, scale=0.5}] coordinates {
    (1,187/342) (2,31/342) (3,31/342) (4,22/342) (5,22/342) (6,18/342) (7,18/342) (8,14/342) (9,14/342) (10,12/342) (11,12/342) (12,11/342) (13,11/342) (14,10/342) (15,10/342) (16,10/342) (17,10/342) (18,10/342)
};
\legend{Test,Train}
\end{axis}
\end{tikzpicture}
\end{minipage}%
\begin{minipage}{.5\linewidth}
  \centering
  \begin{tikzpicture}[scale=0.75]
\begin{axis}[    ymin=7/341, ymax=40/341,    xmin=1, xmax=20,    xlabel={$length$},    ylabel={$error$}, ylabel style={rotate=-90, at={(axis description cs:.08,1.05)}},   axis lines=left,    enlargelimits=true, xtick={1,3,5,7,9,11,13,15,17,19}, yticklabel style={/pgf/number format/.cd, fixed, fixed zerofill, precision=2}]

\addplot[    color=red,    thick, mark=star, mark options={fill=white, scale=0.75}] coordinates {
    (1,191.30/341) (2,36.20/341) (3,36.20/341) (4,22.80/341) (5,22.80/341) (6,22.10/341) (7,22.10/341) (8,20.30/341) (9,20.30/341) (10,15.60/341) (11,15.60/341) (12,15.20/341) (13,15.20/341) (14,14.70/341) (15,15.10/341) (16,15.20/341) (17,15.40/341) (18,15.90/341) (19,15.50/341) (20,16.30/341) 
};
\addplot[    color=black,    thick, mark=*, mark options={fill=white, scale=0.5}] coordinates {
    (1,193.70/342) (2,33.00/342) (3,33.00/342) (4,20.30/342) (5,20.30/342) (6,17.30/342) (7,17.30/342) (8,14.60/342) (9,14.60/342) (10,12.00/342) (11,12.00/342) (12,10.70/342) (13,10.70/342) (14,9.90/342) (15,9.80/342) (16,9.60/342) (17,9.50/342) (18,8.90/342) (19,8.70/342) (20,8.50/342) 
};

\legend{Test,Train}
\end{axis}
\end{tikzpicture}
\end{minipage}
\caption{Breast cancer data set -- first test (left) and average (right)\label{fig:cancer-data}}

\bigskip

\bigskip

\centering
\begin{minipage}{.5\linewidth}
  \centering
\begin{tikzpicture}[scale=0.75]
\begin{axis}[    ymin=13/175, ymax=50/175,    xmin=1, xmax=12,    xlabel={$length$},    ylabel={$error$},  ylabel style={rotate=-90, at={(axis description cs:.08,1.05)}},  axis lines=left,    enlargelimits=true, xtick={1,3,5,7,9,11}]

\addplot[    color=red,    thick,    mark=star, mark options={fill=white, scale=0.75}] coordinates {
    (1,42/175) (2,42/175) (3,38/175) (4,38/175) (5,29/175) (6,29/175) (7,25/175) (8,25/175) (9,40/175) (10,40/175) (11,32/175) (12,32/175)
};
\addplot[    color=black,    thick,    mark=*, mark options={fill=white, scale=0.5}] coordinates {
    (1,46/176) (2,46/176) (3,36/176) (4,36/176) (5,27/176) (6,27/176) (7,23/176) (8,23/176) (9,20/176) (10,20/176) (11,18/176) (12,18/176)
};
\legend{Test,Train}
\end{axis}
\end{tikzpicture}
\end{minipage}%
\begin{minipage}{.5\linewidth}
  \centering
  \begin{tikzpicture}[scale=0.75]
\begin{axis}[    ymin=13/175, ymax=50/175,    xmin=1, xmax=16,    xlabel={$length$},    ylabel={$error$},  ylabel style={rotate=-90, at={(axis description cs:.08,1.05)}},  axis lines=left,    enlargelimits=true, xtick={1,3,5,7,9,11,13,15}]

\addplot[    color=red,    thick, mark=star, mark options={fill=white, scale=0.75}] coordinates {
    (1,44.90/175) (2,44.90/175) (3,41.90/175) (4,43.00/175) (5,32.80/175) (6,32.80/175) (7,31.20/175) (8,33.60/175) (9,31.80/175) (10,30.20/175) (11,29.90/175) (12,28.10/175) (13,28.40/175) (14,26.80/175) (15,29.80/175) 
};
\addplot[    color=black,    thick, mark=*, mark options={fill=white, scale=0.5}] coordinates {
    (1,43.10/176) (2,43.10/176) (3,35.90/176) (4,35.80/176) (5,26.90/176) (6,26.90/176) (7,24.30/176) (8,23.70/176) (9,21.70/176) (10,20.40/176) (11,19.00/176) (12,17.50/176) (13,16.60/176) (14,15.20/176) (15,14.90/176) 
};
\legend{Test,Train}
\end{axis}
\end{tikzpicture}
\end{minipage}
\caption{Ionosphere data set -- first test (left) and average (right)\label{fig:ionosphere-data}}

\medskip

\end{figure}


%
%
%
%
%
%
%

\medskip 

\smallskip

\noindent \textbf{Overfitting and choice of explanations.}
The six plots show how the error rates develop with formula length.
In all plots, the error of the test data eventually stays roughly the same while the error of the training data keeps decreasing.
This illustrates how the overfitting phenomenon ultimately arises.
We can use these results to find a \emph{cut-off point} for the length of the formulas to be used as explanations. Note that this should be done on a case-by-case basis and we show the average plots only to demonstrate trends. For the single tests given on the left in Figures \ref{fig:credit-data}, 
\ref{fig:cancer-data} and 
\ref{fig:ionosphere-data}, 
we might choose the lengths 6, 8 and 7 for the credit, breast cancer and ionosphere data sets, respectively. The errors of the chosen formulas are $0.27$, $0.047$ and $0.14$, respectively. We conclude that by sticking to short Boolean formulas, we can avoid overfitting in a simple way.

\smallskip 

\medskip 

\noindent \textbf{Interpretability.}
A nice feature 
of short Boolean formulas is
their interpretability. Suppose we stop at the formula lengths 6, 8 and 7 suggested above. The related formulas are simple and indeed readable. Consider the formula 

\medskip

{\centering
$\displaystyle \neg(a[1,1] \wedge a[2]) \vee a[17,4]$\par}

\medskip

\noindent of length 6 and a test error of 0.27 obtained from the credit data. The meanings of the attributes are as follows: $a[1,1]$ means the applicant has a checking account with negative balance, $a[2]$ means that the duration of the applied loan is above median, and $a[17,4]$ means the applicant is employed at a management or similar level. (The second number in some attributes refers to original categories in the data.) Therefore the formula states that if an applicant is looking for a short term loan, has money on their account or has a management level job, then they should get the loan. For the breast cancer data set, we choose the formula

\medskip

{\centering
$\displaystyle \neg (((a[1] \wedge a[6]) \vee a[5])  \wedge a[3])$ \par}

\medskip

\noindent of length 8 with test error $0.047$. The meanings of the attributes in the order of appearance in the formula are given as clump thickness, bare nuclei, single epithelial cell size and uniformity of cell shape. 
The full power of Boolean logic is utilized here, in the form of both negation and alternation between conjunction and disjunction. 
Finally, for the ionosphere data set, the formula

\medskip

{\centering
$\displaystyle 
((a[8] \wedge a[12]) \vee a[15]) \wedge a[1]$ \par}

\medskip

\noindent of length 7 and test error $0.14$ is likewise human readable as a formula. However, it must be mentioned again that the data was used here for technical reasons, and the Booleanized
attributes related to radar readings are
difficult to explicate.

Using the power of Boolean logic (i.e., including all the connectives $\neg$, $\wedge$, $\vee$) tends to compress the explanations suitably in addition to giving flexibility in explanations. 
We observe that our experiments gave short, readable formulas. 

\medskip 

\smallskip 


\noindent \textbf{Comparing error rates on test data.}
In \cite{YangW01}, all three data sets we consider are treated with \emph{naive Bayesian classifiers} and error rates  
$0.25$, $0.026$ and $0.10$ are achieved on the test data for the credit, breast cancer and ionosphere data sets, respectively. In \cite{Griffith2003}, the credit data is investigated using neural networks and even there, the best reported error rate 
is $0.24$. In \cite{SterD96}, many different methods are compared on the breast cancer data, and the best error achieved is $0.032$. For the ionosphere data, the original paper \cite{Sigillito} uses neural networks to obtain an error of $0.04$. We can see from the plots that 
very short Boolean formulas can achieve error rates of a similar magnitude on the credit and breast cancer data sets. For the ionosphere data, neural networks
achieve a better error rate, but as explained earlier, this is unsurprising as we use a roughly Booleanized version of the
underlying data. 
We conclude that very short Boolean formulas give surprisingly good error rates compared to other methods. Furthermore, this approach seems inherently interpretable for many different purposes. 

\section{Polynomial time algorithm for learning}

Our empirical results indicate that already a few key attributes seem to suffice for competent explanations. Based on this observation, we briefly envision here, as a future direction, the study and use of the following algorithm. Given a method to choose a small set of promising attributes, we simply compute the \emph{ideal classifier} with respect to those attributes. The resulting formula, while not of minimal length in general, is reasonably small due to the small number of attributes used. One can also minimize the obtained formula at the end to enhance interpretability.

Algorithm \ref{algo:overfitter} describes a more detailed implementation of the learning algorithm. Let $\tau$ be a small set of promising attributes chosen from the initially possibly large set of all attributes. The intuitive idea behind the algorithm is as follows. For every $\tau$-type $t$ that is realized in $M$, we have two counters, $n_t$ and $c_t$. The first counter $n_t$ counts how many times $t$ is realized in $M$, while $c_t$ counts how many times $t \land q$ is realized in $M$. The number $c_t/n_t$ is then the probability $\Pr_M(q|t)$. The ideal classifier $\varphi_{id}^{M}$ can be constructed by taking a disjunction over all the types $t$ for which $c_t/n_t \geq 1/2$.

\begin{breakablealgorithm}
\caption{Compute the ideal classifier $\varphi^M_{id}$}\label{algo:overfitter}
\hspace*{\algorithmicindent} \textbf{Input:} a $(\tau \cup \{q\})$-model $M$
\begin{algorithmic}[1]
\State $T_M \leftarrow \varnothing$ \Comment{All the $\tau$-types realized in $M$ will be stored in the set $T_M$}
\State \textbf{for} $w \in W$ \textbf{do}
\State \quad \quad $t \leftarrow$ the $\tau$-type of $w$
\State \quad \quad \textbf{if} $t \not\in T_M$ \textbf{then}
\State \quad \quad \quad \quad $T_M \leftarrow t$
\State \quad \quad \quad \quad $n_t, c_t \leftarrow 0$
\State \quad \quad $n_t \leftarrow n_t + 1$
\State \quad \quad \textbf{if} $w \in q$ \textbf{then}
\State \quad \quad \quad \quad $c_t \leftarrow c_t + 1$
\State $\varphi_{id}^M \leftarrow \bot$
\State \textbf{for} $t \in T_M$ \textbf{do}
\State \quad \quad \textbf{if} $c_t / n_t \geq 1/2$ \textbf{then}
\State \quad \quad \quad \quad $\varphi_{id}^M \leftarrow \varphi_{id}^M \lor t$
\State \textbf{return $\varphi_{id}^M$}
\end{algorithmic}
\end{breakablealgorithm}
It is clear that this algorithm runs in polynomial time with respect to the size of $M$, which is $\mathcal{O}(|W||\tau|)$. A more precise analysis shows that the running time of this algorithm is $\mathcal{O}(|W||T_M||\tau|),$ where $|T_M|$ counts the number of $\tau$-types that are realized in $M$. Since $|T_M| \leq |W|$, this gives a worst case of $\mathcal{O}(|W|^2|\tau|)$, when every point realizes a different type. If $\tau$ is thought of as being of fixed size, then this is a linear time algorithm.

Notice that if the initial data set $M$ is large enough, the ideal classifier with respect to $M$ is most likely also the ideal classifier with respect to the underlying probability distribution from which $M$ was sampled. As discussed in Section \ref{section:expected-errors}, how large $M$ needs to be depends in an essential way on how many attributes we use. By using few enough attributes, we can expect that the ideal classifier obtained from $M$ will be the true ideal classifier.

\section{Explaining a classifier}\label{sec:explaining-a-classifier}

In this section we discuss how our method can be adapted for explaining the behavior of possibly black box classifiers.

First, we consider explaining the complete Boolean behaviour of a classifier $f: T_\tau \to \{0,1\}$. That is, we want to find a short Boolean formula $\varphi$ that behaves identically to $f$ with respect to most propositional types. For this, we generate a target attribute $q$ by feeding the classifier $f$ one input of each propositional type and recording the decision as $q$. We then compute formulas $\varphi$ of increasing length with minimal error as we did above for target attributes $q$ in data. As before, longer formulas will achieve smaller error, up to the (generally very long) formula that is equivalent to $f$. Note that overfitting is not an issue here as we use all propositional types and we are seeking to explain the precise behaviour of the classifier, rather than an underlying phenomenon to be classified.

As a second case, we consider explaining the classifier $f$ with respect to a probability distribution $\mu : T_\tau \to [0,1]$ of inputs. Note that in the extreme scenario where a formula $\varphi$ is equivalent to $f$, the formula $\varphi$ explains $f$ with respect to any distribution, but this kind of $\varphi$ is likely to be very long. To find shorter explanations, we settle for formulas $\varphi$ that behave similarly to $f$ with respect to inputs from the distribution $\mu$. If we somehow know the distribution $\mu$ beforehand, we can produce a finite sample of inputs that follows $\mu$ and feed that to the classifier $f$ to obtain $q$. More relevantly, if we do not know the distribution $\mu$ exactly, we use data sampled from $\mu$ to obtain $q$ via $f$. After obtaining $q$, we search for formulas $\varphi$ as before. In this case, overfitting can be an issue so cross validation should be used.

Finally we compare short explanations $\varphi$ obtained for $f$ with respect to a distribution $\mu$ to the short explanations $\psi$ obtained via the uniform distribution. While over the distribution $\mu$, the explanations $\varphi$ are generally more accurate, the explanations $\psi$ should nevertheless be most accurate on average over all distributions.

\section{Conclusion}
\label{section:conclusion}

We have studied short Boolean formulas as a platform for
producing explanations and interpretable classifiers. 
We have investigated the theoretical reasons behind 
overfitting and provided related quantitative bounds. 
Also, we have tested the approach with three different 
data sets, where the resulting formulas are indeed
interpretable---all being 
genuinely short---and relatively accurate. 
In general, short formulas may
sometimes be necessary to avoid overfitting, and moreover, shorter length leads to increased interpretability.

Our approach need not limit to 
Boolean formulas only, as we can naturally 
extend our work to general relational data. We can use, e.g., description logics and compute
concepts $C_1,\dots , C_k$ and then perform our
procedure using $C_1,\dots, C_k$, finding 
short Boolean combinations of concepts. This of
course differs from the approach of computing
minimal length formulas in the original description logic, but
can nevertheless be fruitful and interesting. We leave this for
future work. Further future directions include, e.g., knowledge discovery via computing all formulas 
up to some short length $\ell$ with
errors smaller than a given threshold.

\medskip

\medskip

\medskip

\noindent
\textbf{Acknowledgments.}\ \ \ 
Tomi Janhunen, Antti Kuusisto, Masood Feyzbakhsh Rankooh and Miikka Vilander were supported by the Academy of Finland consortium project \emph{Explaining AI via Logic} (XAILOG), grant numbers 345633 (Janhunen) and 345612 (Kuusisto). Antti Kuusisto was also supported by the Academy of Finland project \emph{Theory of computational logics}, grant numbers 324435, 328987 (to December 2021) and 352419, 352420 (from January 2022). The author names of this article have been ordered on the basis of alphabetic order. 


\bibliographystyle{plain}
\bibliography{local}


\newpage
\section{Appendix}

\subsection{The remaining part of
the proof of Proposition \ref{errorproposition}}

\begin{proof}
Let $\mathcal{D}$ denote the set of
distributions $\mu : T_{\tau,q}\rightarrow \mathbb{R}$. 
Define the random variable $T : \mathcal{D}\rightarrow \mathbb{R}$
such that $T(t)$ gives the true error of $t$ with respect to $\mu$. We should
find the expected value $\mathbb{E}(T)$ of $T$ with
respect to the flat Dirichlet distribution over $\mathcal{D}$. 
To this end, we fix
the following random 
variables of type $\mathcal{D} \rightarrow \mathbb{R}$. First, for 
each $s\in T_{\tau}$, 
let $X_s:\mathcal{D}
\rightarrow \mathbb{R}$ be the function 
such that $X_s(\mu) = \mu_{\tau}(s)$, so $X_s$ simply 
gives the probability of $s$. 
Similarly, for each $s\in T_{\tau}$,
define $Z_s:\mathcal{D}\rightarrow \mathbb{R}$ 
such that $Z_s(\mu) = \mathit{cut}(\mu( q\, |\, s))$, i.e.,
letting $p$ denote the conditional
probability of $q$ given $s$ holds, the
output is $p$ if $p\leq \frac{1}{2}$ and otherwise $(1-p)$. 
Now we have 
$\mathbb{E}(T) = \mathbb{E}(\sum_{s\, \in\, T_{\tau}}X_sZ_s) 
= \sum_{s\, \in\, T_{\tau}}\mathbb{E}(X_sZ_s)$,
where we have used the linearity of $\mathbb{E}$. 
Now, since $X_s$ gives the
probability of $s$ and $Z_s$ the error associated with $q$
when $s$ holds, we notice that $X_s$ and $Z_s$ are independent when the 
distributions in $\mathcal{D}$ are
equally likely. 
Due to the independence, we 
have $\mathbb{E}(X_s Z_s) = \mathbb{E}(X_s)\mathbb{E}(Z_s)$.
Therefore we have $\mathbb{E}(T) = \sum_{s\, \in\, T_{\tau}}
\mathbb{E}(X_s) \mathbb{E}(Z_s)$. 
Now, consider some $s\in T_{\tau, q}$. 
First, $\mathbb{E}(X_s) = \frac{1}{|T_{\tau}|}$.
Secondly, if we let $x$ 
denote $\mu( q\, |\, s) = \frac{\mu(q \wedge s)}{\mu(s)}$, 
then $\mathbb{E}(Z_s)$ is the 
expected value $\frac{1}{1-0}\int_0^1\mathit{cut}(x)dx$ of $\mathit{cut}(x)$, which we 
have already computed to be $\frac{1}{4}$. 
Therefore we have $\mathbb{E}(T) = \sum_{s\, \in\, T_{\tau}}
\frac{1}{|T_{\tau}|}\cdot \frac{1}{4}
= \frac{1}{4}$.  
\end{proof}

\subsection{Proof that (\ref{eq:conditioned_expected_empirical_error}) equals (\ref{eq:simplified_expected_empirical_error})}

\begin{proof}
    We want to establish the following identity
    \begin{align*}
        &\frac{1}{n} \sum_{t \in T_\tau} \sum_{m = 0}^n \mathbb{E}(\min\{|t\land q|_M,m - |t\land q|_M\} \mid |t|_M = m) \cdot \Pr(|t|_M = m) \\
        &= \frac{1}{n} \sum_{t \in T_\tau} \sum_{m = 0}^n \mathbb{E}\min\{X_{m,t,q}, m - X_{m,t,q}\} \cdot \Pr(|t|_M = m)
    \end{align*}
    Since
    \begin{align*}
    &\mathbb{E}(\min\{|t\land q|_M,m - |t\land q|_M\} \mid |t|_M = m) \\
    &= \sum_{k = 0}^n \min\{k,m - k\} \cdot \Pr(|t \land q|_M = k \mid |t|_M = m) \\
    &= \sum_{k = 0}^m \min\{k,m - k\} \cdot \Pr(|t \land q|_M = k \mid |t|_M = m)
    \end{align*}
    it suffices to show that
    \[\Pr(|t \land q|_M = k \mid |t|_M = m) = \Pr(X_{m,t,q} = k),\]
    for every $0 \leq k \leq m$. First, we have that
    \[\Pr(X_{m,t,q} = k) := \binom{m}{k} \mu(q|t)^k (1 - \mu(q|t))^{m-k}\]
    Recalling that $|t \land q|_M$ is distributed according to $\mathrm{Binom}(n,\mu(q \land t))$, we can calculate $\Pr(|t \land q|_M = k \mid |t|_M = m)$ as follows.
    \begin{align*}
        \Pr(|t \land q|_M = k \mid |t|_M = m) 
        &:= \frac{\Pr(|t \land q|_M = k \text{ and } |t|_M = m)}{\Pr(|t|_M = m)} \\
        &=  \frac{\binom{n}{m} \binom{m}{k} \mu(t \land q)^k \mu(t \land \neg q)^{m-k} (1 - \mu(t))^{n - m}}{\binom{n}{m} \mu(t)^m (1 - \mu(t))^{n-m}} \\
        &= \frac{\binom{m}{k} \mu(t \land q)^k \mu(t \land \neg q)^{m - k}}{\mu(t)^m} \\
        &= \binom{m}{k} \bigg(\frac{\mu(t\land q)}{\mu(t)}\bigg)^k \bigg(\frac{\mu(t\land \neg q)}{\mu(t)}\bigg)^{m - k}
    \end{align*}
    Since $\mu(t \land q)/\mu(t) = \mu(q|t)$ and  $\mu(t \land \neg q)/\mu(t) = 1 - \mu(q|t)$, we are done.  
\end{proof}

\subsection{Proof that (\ref{eq:simplified_expected_empirical_error}) equals (\ref{eq:alternative_expected_empirical_error_formula})}

\begin{proof}
We want to establish the following identity:
\begin{align*}
    &\frac{1}{n} \sum_{t \in T_\tau} \sum_{m = 0}^n \mathbb{E}\min\{X_{m,t,q}, m - X_{m,t,q}\} \cdot \Pr(|t|_M = m) \\
    &= \frac{1}{2} - \frac{1}{n} \sum_{t \in T_\tau} \sum_{m = 0}^n \mathbb{E}\bigg|X_{m,t,q} - \frac{m}{2}\bigg| \cdot \Pr(|t|_M = m)
\end{align*}
Applying the following identity
\[\min\{a,b\} = \frac{1}{2}(a + b - |a - b|)\]
on (\ref{eq:simplified_expected_empirical_error}) gives us the following chain of equalities
\begin{align*}
    &\frac{1}{n} \sum_{t \in T_\tau} \sum_{m = 0}^n \mathbb{E}\min\{X_{m,t,q}, m - X_{m,t,q}\} \cdot \Pr(|t|_M = m) \\
    &= \frac{1}{n} \sum_{t \in T_\tau} \sum_{m = 0}^n \bigg(\frac{1}{2} m - \mathbb{E}\bigg|X_{m,t,q} - \frac{m}{2}\bigg|\bigg) \cdot \Pr(|t|_M = m) \\
    &= \frac{1}{2n} \sum_{t \in T_\tau} \sum_{m = 0}^n m \cdot \Pr(|t|_M = m) - \frac{1}{n} \sum_{t \in T_\tau} \sum_{m = 0}^n \mathbb{E}\bigg|X_{m,t,q} - \frac{m}{2}\bigg| \cdot \Pr(|t|_M = m)
\end{align*}
Now observe that $\sum_{m = 0}^n m \cdot \Pr(|t|_M = m)$ is simply the expected value of $|t|_M$, which is just $n\mu(t)$, since $|t|_M$ was distributed according to $\mathrm{Binom}(n,\mu(t))$. This, together with the observation that $\sum_{t \in T_\tau} \mu(t) = 1$, entails that the previous formula simplifies to (\ref{eq:alternative_expected_empirical_error_formula}).
\end{proof}

\subsection{Proof of Theorem \ref{thm:expected_empirical_error_lower_bound}}

\begin{proof}
    Jensen's inequality implies that
    \[\mathbb{E}\bigg|X_{m,t,q} - \mathbb{E}(X_{m,t,q})\bigg| \leq \sqrt{\mathrm{Var}(X_{m,t,q})} = \sqrt{m\mu(q|t)(1 - \mu(q|t))},\]
    since $X_{m,t,q}$ was distributed according to $\mathrm{Binom}(m,\mu(t))$. Using this together with triangle inequality we obtain that
    \begin{align*}
    \mathbb{E}\bigg|X_{m,t,q} - \frac{m}{2}\bigg| 
    &\leq \mathbb{E}\bigg|X_{m,t,q} - \mathbb{E}(X_{m,t,q})\bigg| + \mathbb{E}\bigg|\mathbb{E}(X_{m,t,q}) - \frac{m}{2}\bigg| \\
    &\leq \sqrt{m\mu(q|t)(1 - \mu(q|t))} + \bigg|\mu(q|t) - \frac{1}{2}\bigg|m
    \end{align*}
    Thus we have that
    \begin{align*}
    & \frac{1}{n}\sum_{t \in T_\tau} \sum_{m=0}^n \mathbb{E}\bigg|X_{m,t,q} - \frac{m}{2}\bigg| \cdot \Pr(|t|_M = m) \\
    &\leq \frac{1}{n} \sum_{t \in T_\tau} \sqrt{\mu(q|t)(1 - \mu(q|t))} \underbrace{\sum_{m = 0}^n \sqrt{m} \Pr(|t|_M = m)}_{= \mathbb{E}(\sqrt{|t|_M})} \\
    &+ \frac{1}{n} \sum_{t \in T_\tau} \bigg|\mu(q|t) - \frac{1}{2}\bigg| \underbrace{\sum_{m = 0}^n m \Pr(|t|_M = m)}_{= \mathbb{E}(|t|_M) = n\mu(t)}\\
    &= \frac{1}{n} \sum_{t \in T_\tau} \sqrt{\mu(q|t)(1 - \mu(q|t))} \mathbb{E}(\sqrt{|t|_M}) + \sum_{t \in T_\tau} \bigg|\mu(q|t) - \frac{1}{2}\bigg| \mu(t)
    \end{align*}
    Using again Jensen's inequality, we see that $\mathbb{E}(\sqrt{|t|_M}) \leq \sqrt{\mathbb{E}(|t|_M)} = \sqrt{n\mu(t)}$. Hence, we can bound (\ref{eq:alternative_expected_empirical_error_formula}) from below by
    \[\bigg(\frac{1}{2} - \sum_{t \in T_\tau} \bigg|\mu(q|t) - \frac{1}{2}\bigg|\mu(t)\bigg) - \frac{1}{\sqrt{n}} \sum_{t \in T_\tau} \sqrt{\mu(q|t)(1 - \mu(q|t)) \mu(t)}\]
    To conclude the proof, we note that
    \begin{align*}
    \frac{1}{2} - \sum_{t \in T} \bigg|\mu(q|t) - \frac{1}{2}\bigg|\mu(t)
    &= \sum_{t \in T} \frac{1}{2} \mu(t) - \sum_{t \in T} \bigg|\mu(q|t) - \frac{1}{2}\bigg|\mu(t) \\
    &= \sum_{t \in T} \bigg(\frac{1}{2} - \bigg|\mu(q|t) - \frac{1}{2}\bigg|\bigg)\mu(t) \\
    &= \sum_{t \in T} \min\{\mu(q|t),1-\mu(q|t)\} \mu(t) \\
    &= \sum_{t \in T} \min\{\mu(t \land q), \mu(t \land \neg q)\}
    \end{align*}  
\end{proof}

\subsection{Results required for the proof of Theorem \ref{thm:upper_bound_expected_empirical_error}} 
$\empty{}$ \\
 We start with the following proposition which implies that (\ref{eq:simplified_expected_empirical_error}) (and hence $\mathbb{E}(\mathrm{err}(M))$) is maximized when $\mu(q|t) = 1/2$, for every $t \in T_\tau$.

\begin{proposition}
    Suppose that $X \sim \mathrm{Binom}(n,1/2), Y \sim \mathrm{Binom}(n,p)$. Then
    \[\mathbb{E} \min\{Y, n - Y\} \leq \mathbb{E}\min\{X,n-X\}\]
\end{proposition}
\begin{proof}
    For simplicity we will assume that $n$ is odd. Now we can rewrite the left-hand side of the claimed inequality as follows.
    \begin{align*}
        \mathbb{E} \min\{Y, n - Y\}
        &= \sum_{k = 0}^n \min\{k, n - k\} \cdot \Pr(Y = k) \\
        &= \sum_{k = 0}^{\lfloor n/2 \rfloor} \min\{k, n - k\} \cdot (\Pr(Y = k) + \Pr(Y = n - k)) \\
        &= \sum_{k = 0}^{\lfloor n/2 \rfloor} \min\{k, n - k\} \cdot \binom{n}{k} \cdot (p^k(1 - p)^{n - k} + p^{n - k} (1 - p)^k)
    \end{align*}
    The function $p^k(1 - p)^{n-k} + p^{n - k}(1 - p)^k$ is maximized on the interval $[0,1]$ when $p = 1/2$. Hence
    \begin{align*}
        &\sum_{k = 0}^{\lfloor n/2 \rfloor} \min\{k, n - k\} \cdot \binom{n}{k} \cdot (p^k(1 - p)^{n - k} + p^{n - k} (1 - p)^k) \\
        &\leq  \sum_{k = 0}^{\lfloor n/2 \rfloor} \min\{k, n - k\} \cdot \binom{n}{k} \cdot \bigg(\bigg(\frac{1}{2}\bigg)^n + \bigg(\frac{1}{2}\bigg)^n\bigg) \\
        &= \sum_{k = 0}^{\lfloor n/2 \rfloor} \min\{k, n - k\} \cdot (\Pr(X = k) + \Pr(X = n - k)) \\
        &= \sum_{k = 0}^n \min\{k, n - k\} \cdot \Pr(X = k) \\
        &= \mathbb{E} \min\{X,n-X\},
    \end{align*}
    which is what we wanted to show.  
\end{proof}

\begin{lemma}\label{lemma:key_estimate}
    Fix a positive integer $n \geq 1$ and let $X \sim \mathrm{Binom}(n,1/2)$. Then
    \[\mathbb{E}\bigg|X - \frac{n}{2}\bigg| \geq \sqrt{\frac{n}{8}}\]
\end{lemma}
\begin{proof}
    For $n = 1$ a straightforward calculation shows that
    \[\mathbb{E}\bigg|X - \frac{n}{2}\bigg| = \frac{1}{2} \geq \sqrt{\frac{1}{8}}\]
    For $n \geq 2$ this is Lemma 2.3 in \cite{PELEKIS2016305}.  
\end{proof}

\begin{lemma}\label{lemma:square-one}
    Let $X \sim \mathrm{Binom}(n,p)$. Then
    \[\mathbb{E}(\sqrt{X}) \geq \sqrt{np} - \frac{1 - p}{2\sqrt{np}}\]
\end{lemma}
\begin{proof}
    We will follow the argument sketched in \cite{MathOverFlowAnswer}. First, we start with the following inequality.
    \[\sqrt{x} \geq 1 + \frac{x - 1}{2} - \frac{(x - 1)^2}{2},\]
    A direct computation shows that this inequality is valid for all $x \geq 0$. We note that this inequality is suggested by looking at the Taylor series of $\sqrt{x}$ around $x = 1$, where the first three terms are
    \[1 + \frac{x - 1}{2} - \frac{(x - 1)^2}{8}\]
    Now we replace $x$ with $\frac{X}{\mathbb{E}(X)}$ in this inequality to get 
    \[\mathbb{E}(\sqrt{X}) \geq \sqrt{\mathbb{E}(X)}\bigg(1 - \frac{\mathrm{Var}(X)}{2\mathbb{E}(X)^2}\bigg)\]
    Since $X \sim \mathrm{Binom}(n,p)$ we know that $\mathbb{E}(X) = np$ and that $\mathrm{Var}(X) = np(1 - p)$. Plugging these values in the above inequality yields the desired result.  
\end{proof}

The above lemma gives meaningful estimates only when $p$ is not too small when compared to $n$. When $pn < 1$ we will use the following alternative bound.

\begin{lemma}\label{lemma:square-two}
    Let $X \sim \mathrm{Binom}(n,p)$. Then
    \[\mathbb{E}(\sqrt{X}) \geq np(1 - p)^n\]
\end{lemma}
\begin{proof}
    \begin{align*}
        \mathbb{E}(\sqrt{X}) = \sum_{k = 0}^n \sqrt{k} \Pr(X = k) \geq \sqrt{1} \cdot \Pr(X = 1) = np(1-p)^n
    \end{align*}  
\end{proof}

\subsection{Proof of Theorem \ref{thm:upper_bound_expected_empirical_error}}

\begin{proof}
    With our assumption we know that $X_{m,t,q} \sim \mathrm{Binom}(m,1/2)$. Lemma \ref{lemma:key_estimate} implies that
    \[\mathbb{E}\bigg|X_{m,t,q} - \frac{m}{2}\bigg| \geq \sqrt{\frac{m}{8}}\]
    Thus we have that
    \begin{align*}
    &\frac{1}{n} \sum_{t \in T_\tau} \sum_{m = 0}^n \mathbb{E}\bigg|X_{m,t,q} - \frac{m}{2}\bigg| \cdot \Pr(|t|_M = m) \\
    &\geq \frac{1}{n\sqrt{8}} \sum_{t \in T_\tau} \sum_{m = 0} \sqrt{m} \cdot \Pr(|t|_M = m) \\
    &= \frac{1}{n\sqrt{8}} \sum_{t \in T_\tau} \mathbb{E}(\sqrt{|t|_M})
    \end{align*}
    By splitting the sum $\sum_{t \in T_\tau} \mathbb{E}(\sqrt{|t|_M})$ into two parts and using lemmas \ref{lemma:square-one} and \ref{lemma:square-two} respectively we can bound the last line from below using
    \[\frac{1}{\sqrt{8n}}\sum_{n\mu(t) \geq 1} \sqrt{\mu(t)} - \frac{1}{2\sqrt{8}n}\sum_{n\mu(t) \geq 1} \frac{1 - \mu(t)}{\sqrt{n\mu(t)}} + \frac{1}{\sqrt{8}} \sum_{n\mu(t) < 1} \mu(t)(1 - \mu(t))^n,\]
    which concludes the proof. 
\end{proof}

\end{document}